\documentclass{acm_proc_article-sp}
\usepackage{caption}
\usepackage{graphicx} % more modern
\usepackage{xspace}
\usepackage{url}
\usepackage{subfigure}
\usepackage{epsfig}

\usepackage{algorithmic}
\usepackage{algorithm}
\usepackage{mathrsfs}
\makeatletter
 \let\@copyrightspace\relax
 \makeatother
 
\usepackage{amssymb}
\usepackage{amsmath}
\usepackage{amsfonts}

\numberwithin{equation}{section}
\numberwithin{figure}{section}

\newtheorem{thm}{Theorem}

\newtheorem{lemma}[thm]{Lemma}
\newtheorem{corol}[thm]{Corollary}

\newtheorem{defn}[thm]{Definition}
\newtheorem{rem}[thm]{Remark}
\newtheorem{fact}[thm]{Fact}

\newcommand{\comment}[1]{}

\newcommand{\algo}{Layered LSH}
\newcommand*\samethanks[1][\value{footnote}]{\footnotemark[#1]}

\begin{document}
\title{Efficient Distributed Locality Sensitive Hashing}
\numberofauthors{3}
\author{
\alignauthor
Bahman Bahmani \\
\affaddr{Stanford Univesrity \\ Stanford, CA \\bahman@stanford.edu}
\thanks{Research supported by NSF grant 0904314}
\thanks{We also acknowledge financial support from grant \#FA9550-12-1-0411}
\and
\alignauthor
Ashish Goel\\
\affaddr{Stanford Univesrity \\ Stanford, CA \\ashishg@stanford.edu}
\thanks{Research supported by NSF grants 0915040 and NSF 0904314}
\samethanks[2]
\and
\alignauthor
Rajendra Shinde\\
\affaddr{Stanford Univesrity \\ Stanford, CA \\rbs@stanford.edu}
\thanks{Research supported by NSF grant 0915040}
\samethanks[2]
}
\maketitle
\begin{abstract}
Distributed frameworks are gaining increasingly widespread use in applications that process large amounts of data. One important example application is large scale similarity search, for which Locality Sensitive Hashing (LSH) has emerged as the method of choice, specially when the data is high-dimensional. At its core, LSH is based on hashing the data points to a number of buckets such that similar points are more likely to map to the same buckets. To guarantee high search quality, the LSH scheme needs a rather large number of hash tables. This entails a large space requirement, and in the distributed setting, with each query requiring a network call per hash bucket look up, this also entails a big network load. The Entropy LSH scheme proposed by Panigrahy significantly reduces the number of required hash tables by looking up a number of query offsets in addition to the query itself. While this improves the LSH space requirement, it does not help with (and in fact worsens) the search network efficiency, as now each query offset requires a network call. In this paper, focusing on the Euclidian space under $l_2$ norm and building up on Entropy LSH, we propose the distributed Layered LSH scheme, and prove that it exponentially decreases the network cost, while maintaining a good load balance between different machines. Our experiments also verify that our scheme results in a significant network traffic reduction that brings about large runtime improvement in real world applications.
\end{abstract}

\section{Introduction}
\label{sec:intro}

Similarity search is the problem of retrieving data objects similar to a query object. It has become an important component of modern data-mining systems, with applications ranging from de-duplication of web documents, content-based audio, video, and image search \cite{Brian:Kulis, Charikar:multiprobe, google:video:lsh}, collaborative filtering \cite{google-news}, large scale genomic sequence alignment \cite{Buhler}, natural language processing \cite{ravi:NLP:clustering}, pattern classification \cite{cover67}, and clustering \cite{berkhin}.  
%\textbf{NEED MANY MORE REFERENCES FOR EACH APPLICATION}

In these applications, objects are usually represented by a high dimensional feature vector. A scheme to solve the similarity search problem constructs an index which, given a query point, allows for quickly finding the data points similar to it. In addition to the query search procedure, the index construction also needs to be time and space efficient. Furthermore, since today's massive datasets are typically stored and processed in a distributed fashion, where network communication is one of the most important bottlenecks, these methods need to be network efficient, as otherwise, the network load would slow down the whole scheme.

An important family of similarity search methods is based on the notion of Locality Sensitive Hashing (LSH) \cite{im98}. At its core, LSH is based on hashing the (data and query) points into a number of hash buckets such that similar points have higher chances of getting mapped to the same buckets. Then for each query, the nearest neighbor among the data points mapped to a same bucket as the query point is returned as the search result.

LSH has been shown to scale well with the data dimension \cite{im98,KOR98}. However, the main drawback of conventional LSH based schemes is that to guarantee a good search quality, one needs a large number of hash tables. This entails a rather large space requirement for the index, and also in the distributed setting, a large network load, as each hash bucket look up requires a communication over the network. To mitigate the space efficiency issue, Panigrahy \cite{P06} proposed the Entropy LSH scheme, which significantly reduces the number of required hash tables, by looking up a number of query offsets in addition to the query itself. Even though this scheme improves the LSH space efficiency, it does not help with its network efficiency, as now each query offset lookup requires a network call. In fact, since the number of required offsets in Entropy LSH is larger than the number of required hash tables in conventional LSH, Entropy LSH amplifies the network inefficiency issue.

In this paper, focusing on the Euclidian space under $l_2$ norm and building up on the Entropy LSH scheme, we design the Layered LSH method for distributing the hash buckets over a set of machines which leads to a very high network efficiency. We prove that, compared to a straightforward distributed implementation of LSH or Entropy LSH, our Layered LSH method results in an exponential improvement in the network load (from polynomial in $n$, the number of data points, to sub-logarithmic in $n$), while maintaining a good load balance between the different machines. Our experiments also verify that our scheme results in large network traffic improvement that in turn results in significant runtime speedups. 

In the rest of this section, we first provide some background on the similarity search problem and the relevant methods, then discuss LSH in the distributed computation model, and finally present an overview of our scheme as well as our results.

\subsection{Background}
\label{sec:bckgrnd}

In this section, we briefly review the similarity search problem, the basic LSH and Entropy LSH approaches to solving it, the distributed computation framework and its instantiations such as MapReduce and Active DHT, and a straightforward implementation of LSH in the distributed setting as well as its major drawback.

\noindent \textbf{Similarity Search:} The similarity search problem is that of finding data objects similar to a query object. In many practical applications, the objects are represented by multidimensional feature vectors, and hence the problem reduces to finding objects close to the query object under the feature space distance metric. The goal in all these problems is to construct an index, which given the query point, allows for quickly finding the search results. The index construction and the query search both need to be space, time, and network efficient.

\noindent \textbf{Basic LSH:} A method to solve the similarity search problem over high dimensional large datasets is based on a specific type of hash functions, namely Locality Sensitive Hash (LSH) functions, proposed by Indyk and Motwani \cite{im98}. An LSH function maps the points in the feature space to a number of buckets in a way that similar points map to the same buckets with a high chance. Then, a similarity search query can be answered by first hashing the query point and then finding the close data points in the same bucket as the one the query is mapped to. To guarantee both a good search quality and a good search efficiency, one needs to use multiple LSH functions and combine their results. Then, although this approach yields a significant improvement in the running time over both the brute force linear scan and the space partitioning approaches \cite{gim99, Bentley:kdtree, Kakade:covertrees, Krauthgamer:lee:navigating:nets, Rtrees, SRtrees}, unfortunately the required number of hash functions is usually large \cite{Buhler, gim99}, and since each hash table has the same size as the dataset, for large scale applications, this entails a very large space requirement for the index. Also, in the distributed setting, since each hash table lookup at query time corresponds to a network call, this entails a large network load which is also undesirable.

\noindent \textbf{Entropy LSH:} To mitigate the space inefficiency of LSH, Panigrahy \cite{P06} introduced the Entropy LSH scheme. This scheme uses the same hash functions and indexing method as the basic LSH scheme. However, it uses a different query time procedure: In addition to hashing the query point, it hashes a number of query offsets as well and also looks up the hash buckets that any of these offsets map to. The idea is that the close data points are very likely to be mapped to either the same bucket as the query point or to the same bucket as one of the query offsets. This significantly reduces the number of hash tables required to guarantee the search quality and efficiency. Hence, this scheme significantly improves the index space requirement compared to the basic LSH method. However, it unfortunately does not help with the query network efficiency, as each query offset requires a network call. Indeed, since one can see that \cite{im98, P06, Charikar:multiprobe} the number of query offsets required by Entropy LSH is larger than the number of hash tables required by basic LSH, the query network efficiency of Entropy LSH is even worse than that of the basic LSH.

In this paper, we focus on the network efficiency of LSH in distributed frameworks. Two main instantiations of such frameworks are the batched processing system MapReduce \cite{mapreduce} (with its open source implementation Apache Hadoop \cite{Hadoop}), and the real-time processing system denoted as Active Distributed Hash Table (Active DHT), such as Twitter Storm \cite{Storm}. The common feature in all these systems is that they process data in the form of (Key, Value) pairs, distributed over a set of machines. This distributed (Key, Value) abstraction is all we need for both our scheme and analyses to apply. However, to make the later discussions more concrete, here we briefly overview the mentioned distributed systems.

\noindent \textbf{MapReduce:} MapReduce \cite{mapreduce} is a simple model for batched distributed processing using a number of commodity machines, where computations are done in three phases. %By automatically handling the lower level issues such as fault tolerance, it provides a simple computational abstraction, where computations are done in three phases.
The Map phase reads a collection of (Key, Value) pairs from an input source, and by invoking a user defined Mapper function on each input element independently and in parallel, emits zero or more (Key, Value) pairs associated with that input element. The Shuffle phase then groups together all the Mapper-emitted (Key, Value) pairs sharing the same Key, and outputs each distinct group to the next phase. The Reduce phase invokes a user-defined Reducer function on each distinct group, independently and in parallel, and emits zero or more values to associate with the group's Key. The emitted (Key, Value) pairs can then be written on the disk or be the input of the Map phase in a following iteration.

\noindent \textbf{Active DHT:} A DHT (Distributed Hash Table) is a distributed (Key, Value) store which allows Lookups, Inserts, and Deletes on the basis of the Key. The term {\em Active} refers to the fact that an arbitrary User Defined Function (UDF) can be executed on a (Key, Value) pair in addition to Insert, Delete, and Lookup. Twitter's Storm~\cite{Storm} is an example of Active DHT that is gaining widespread use. The Active DHT model is broad enough to act as a distributed stream processing system and as a continuous version of MapReduce \cite{OnePassMR}. All the (Key, Value) pairs in a node of the active DHT are usually stored in main memory to allow for fast real-time processing of data and queries. 

In addition to the typical performance measures of total running time and total space, two other measures are very important for both MapReduce and Active DHTs. First, the total network traffic generated, that is the shuffle size for MapReduce and the number of network calls for Active DHT, and second, the maximum number of values with the same key; a high value here can lead to the ``curse of the last reducer'' in MapReduce \cite{curselast} or to one compute node becoming a bottleneck in Active DHT. 

%\begin{enumerate}
%\item total network traffic generated, that is the shuffle size for MapReduce and the number of network calls for Active DHT
%\item the maximum number of values with the same key; a high value here can lead to the ``curse of the last reducer'' in MapReduce \cite{curselast} or to one compute node becoming a bottleneck in Active DHT.
%\end{enumerate}

Next, we will briefly discuss a simple implementation of LSH in distributed frameworks. 

\noindent \textbf{A Simple Distributed LSH Implementation:} \\ 
Each hash table associates a (Key, Value) pair to each data point, where the Key is the point's hash bucket, and the Value is the point itself. These (Key, Value) pairs are randomly distributed over the set of machines such that all the pairs with the same Key are on the same machine. This is done implicitly using a random hash function of the Key. For each query, first a number of (Key, Value) pairs corresponding to the query point are generated. The Value in all of these pairs is the query point itself. For basic LSH, per hash table, the Key is the hash bucket the query maps to, and for Entropy LSH, per query offset, the Key is the hash bucket the offset maps to. Then, each of these (Key, Value) pairs gets sent to and processed by the machine responsible for its Key. This machine contains all data points mapping to the same query or offset hash bucket. Then, it can perform a search within the data points which also map to the same Key and report the close points. This search can be done using the UDF in Active DHT or the Reducer in MapReduce.

In the above implementation, the amount of network communication per query is directly proportional to the number of hash buckets that need to be checked. However, as mentioned earlier, this number is large for both basic LSH and Entropy LSH. Hence, in large scale applications, where either there is a huge batch of queries or the queries arrive in real-time at very high rates, this will require a lot of communication, which not only depletes the valuable network resources in a shared environment, but also significantly slows down the query search process. In this paper, we propose an alternative way, called Layered LSH, to implement the Entropy LSH scheme in a distributed framework and prove that it exponentially reduces the network load compared to the above implementation, while maintaining a good load balance between different machines.

\subsection{Overview of Our Scheme}
\label{sec:overview}

At its core, Layered LSH is a carefully designed implementation of Entropy LSH in the distributed (Key, Value) model. The main idea is to distribute the hash buckets such that near points are likely to be on the same machine (hence network efficiency) while far points are likely to be on different machines (hence load balance). 

This is achieved by rehashing the buckets to which the data points and the offsets of query points map to, via an additional layer of LSH, and then using the hashed buckets as Keys. More specifically, each data point is associated with a (Key, Value) pair where Key is the mapped value of LSH bucket containing the point, and Value is the point's hash bucket concatenated with the point itself. Also, each query point is associated with multiple (Key, Value) pairs where Value is the query itself and Keys are the mapped values of the buckets which need to be searched in order to answer this query. 

Use of an LSH to rehash the buckets not only allows using the proximity of query offsets to bound the number of (Key, Value) pairs for each query (thus guaranteeing network efficiency), but also ensures that far points are unlikely to be hashed to the same machine (thus maintaining load balance).

\subsection{Our Results}
\label{sec:results}

Here, we present a summary of our results in this paper:

\begin{enumerate}
%\item We design a new scheme, called Layered LSH, to implement Entropy LSH in the distributed (Key, Value) model.%, and provide its pseudo-code for the two major distributed frameworks MapReduce and Active DHT. 
\item We prove that Layered LSH incurs only $O(\sqrt{\log n})$ network cost per query. This is an exponential improvement over the $O(n^{\Theta(1)})$ query network cost of the simple distributed implementation of both Entropy LSH and basic LSH.  
\item Surprisingly, we prove that, the network efficiency of Layered LSH is independent of the search quality. This is in sharp contrast with both Entropy LSH and basic LSH in which increasing search quality directly increases the network cost. This offers a very large improvement in both network efficiency and hence overall run time in settings which require similarity search with high accuracy. We also present experiments which verify this observation on the MapReduce framework.
\item We prove that despite network efficiency (which requires collocating near points on the same machines), Layered LSH sends points which are only $\Omega(1)$ apart to different machines with high likelihood. This shows Layered LSH hits the right tradeoff between network efficiency and load balance across machines.
\item We present experimental results with Layered LSH on Hadoop, which show it also works very well in practice.

\end{enumerate}

The organization of this paper is as follows. 
In section \ref{sec:prelim}, we study the Basic and Entropy LSH indexing methods. In section \ref{sec:distlsh}, we give the detailed description of Layered LSH, including its pseudocode for the MapReduce and Active DHT frameworks, and also provide the theoretical analysis of its network cost and load balance. We present the results of our experiments on Hadoop in section \ref{sec:exp}, study the related work in section \ref{sec:rel}, and conclude in section \ref{sec:conc}. 

\section{Preliminaries}
\label{sec:prelim}

In section \ref{sec:bckgrnd}, we provided the high-level background needed for this paper. Here, we present the necessary preliminaries in further detail. Specifically, we formally define the similarity search problem, the notion of LSH functions, the basic LSH indexing, and Entropy LSH indexing.

\noindent \textbf{Similarity Search:} %As mentioned in section \ref{sec:bckgrnd}, similarity search often reduces to an instantiation of the NN search problem. Here, we present the formal definition of this latter problem. Let $T$ be the domain of all data and query objects, and $\zeta$ to be the distance metric between the objects. Given an approximation ratio $c>1$, the $c$-ANN problem is that of constructing an index that given any query point $q\in T$, allows for quickly finding a data point $p\in T$ whose distance to $q$, as measured by $\zeta$, is at most $c$ times larger than the distance from $q$ to its nearest data point. One can see that \cite{HP01, im98} this problem can be solved by reducing it to the $(c,r)$-NN problem, in which the goal is to return a data point within distance $cr$ of the query point $q$, given that there exists a data point within distance $r$ of $q$.
As mentioned in section \ref{sec:bckgrnd}, similarity search in a metric space with domain $T$ reduces to the problem more commonly known as the $(c,r)$-NN problem, where given an approximation ratio $c>1$, the goal is to construct an index that given any query point $q\in T$ within distance $r$ of a data point, allows for quickly finding a data point $p\in T$ whose distance to $q$ is at most $cr$.

\noindent \textbf{Basic LSH:} To solve the $(c,r)$-NN problem, Indyk and Motwani \cite{im98} introduced the following notion of LSH functions:

\newcommand{\norm}[1]{{\left\Vert#1\right\Vert}_2}
\newcommand{\probH}[1]{{\bf\mbox{\bf Pr}}_{\mathcal{H}}\left[#1\right]}
\newcommand{\etal}{{\em et al. }}
\begin{defn}
\label{def:LSH}
For the space $T$ with metric $\zeta$,  given distance threshold $r$, approximation ratio $c>1$, and probabilities $p_1>p_2$, a family of hash functions $\mathcal{H} = \{h: T \to U\}$ is said to be a $(r,cr,p_1,p_2)$-LSH family if for all $x,y \in T$, 
\begin{equation}\begin{array}{l}
\text{if } \zeta(x,y) \leq r \text{ then } \probH{h(x)=h(y)} \geq p_1, \\ 
\text{if } \zeta(x,y) \geq cr \text{ then } \probH{h(x)=h(y)} \leq p_2. \\
\end{array}\end{equation}
\end{defn}

Hash functions drawn from $\mathcal{H}$ have the property that near points (with distance at most $r$) have a high likelihood (at least $p_1$) of being hashed to the same value, while far away points (with distance at least $cr$) are less likely (probability at most $p_2$) to be hashed to the same value; hence the name locality sensitive. 

LSH families can be used to design an index for the $(c,r)$-NN problem as follows. First, for an integer $k$, let $\mathcal{H'}= \{H: T \to U^k\}$ be a family of hash functions in which any $H\in \mathcal{H'}$ is the concatenation of $k$ functions in $\mathcal{H}$, i.e., $H=(h_1, h_2, \ldots, h_k)$, where $h_i\in \mathcal{H}$ ($1\leq i\leq k$). Then, for an integer $M$, draw $M$ hash functions from $\mathcal{H'}$, independently and uniformly at random, and use them to construct the index consisting of $M$ hash tables on the data points. With this index, given a query $q$, the similarity search is done by first generating the set of all data points mapping to the same bucket as $q$ in at least one hash table, and then finding the closest point to $q$ among those data points. The idea is that a function drawn from $\mathcal{H'}$ has a very small chance ($p_2^k$) to map far away points to the same bucket (hence search efficiency), but since it also makes it less likely ($p_1^k$) for a near point to map to the same bucket, we use a number, $M$, of hash tables to guarantee retrieving the near points with a good chance (hence search quality).

%Indyk and Motwani \cite{im98} proved the following theorem: 

%\begin{thm} With $n$ data points, having an LSH family $\mathcal{H}$ as in definition \ref{def:LSH}, and choosing $k = O(\log{n})$ and $M = O(n^{1/c})$, the LSH indexing scheme above solves the $(c,r)$-NN problem with constant probability.
%\end{thm}

To utilize this indexing scheme, one needs an LSH family $\mathcal{H}$ to start with. Such families are known for a variety of metric spaces, including the Hamming distance, the Earth Mover Distance, and the Jaccard measure \cite{C02}. Furthermore, Datar et al. \cite{DIIM04} proposed LSH families for $l_p$ norms, with $0\leq p\leq 2$, using $p$-stable distributions. For any $W > 0$, they consider a family of hash functions $\mathcal{H}_W:\{h_{{\bf a},b}: \mathbb{R}^d \to \mathbb{Z} \}$ such that 
$$h_{{\bf a},b}(v) = \lfloor \frac{{\bf a} \cdot v+b}{W}\rfloor $$ 
where ${\bf a} \in \mathbb{R}^d$ is a $d$-dimensional vector each of whose entries are chosen independently from a $p$-stable distribution, and $b\in \mathbb{R}$ is chosen uniformly from $[0,W]$. Further improvements have been obtained in various special settings \cite{AI06}. In this paper, we will focus on the most widely used $p$-stable distribution, i.e., the $2$-stable, Gaussian distribution. For this case, Indyk and Motwani \cite{im98} proved the following theorem: 

\begin{thm} With $n$ data points, choosing $k = O(\log{n})$ and $M = O(n^{1/c})$, the LSH indexing scheme above solves the $(c,r)$-NN problem with constant probability.
\end{thm}

Although Basic LSH yields a significant improvement in the running time over both the brute force linear scan and the space partitioning approaches \cite{expt:space:partitioning, Kakade:covertrees, Krauthgamer:lee:navigating:nets}, unfortunately the required number of hash functions is usually large \cite{Buhler, gim99}, which entails a very large space requirement for the index. Also, in the distributed setting, each hash table lookup at query time corresponds to a network call which entails a large network load.

\noindent \textbf{Entropy LSH: } To mitigate the space inefficiency, Panigrahy \cite{P06} introduced the Entropy LSH scheme. This scheme uses the same indexing as in the basic LSH scheme, but a different query search procedure. The idea here is that for each hash function $H\in \mathcal{H'}$, the data points close to the query point $q$ are highly likely to hash either to the same value as $H(q)$ or to a value very close to that. Hence, it makes sense to also consider as candidates the points mapping to close hash values. To do so, in this scheme, in addition to $q$, several ``offsets" $q+\delta_i$ ($1\leq i\leq L$), chosen randomly from the surface of $B(q,r)$, the sphere of radius $r$ centered at $q$, are also hashed and the data points in their hash buckets are also considered as search result candidates. It is conceivable that this may reduce the number of required hash tables, and in fact, Panigrahy \cite{P06} shows that with this scheme one can use as few as $\tilde{O}(1)$ hash tables. The instantiation of his result for the $l_2$ norm is as follows:

\begin{thm}
\label{thm:ELSH}
For $n$ data points, choosing $k\geq \frac{\log n}{\log (1/p_2)}$ (with $p_2$ as in Definition \ref{def:LSH}) and $L = O(n^{2/c})$, as few as $\tilde{O}(1)$ hash tables suffice to solve the $(c,r)$-NN problem.
\end{thm}

Hence, this scheme in fact significantly reduces the number of required hash tables (from $O(n^{1/c})$ for basic LSH to $\tilde{O}(1)$), and hence the space efficiency of LSH. However, in the distributed setting, it does not help with reducing the network load of LSH queries. Actually, since for the basic LSH, one needs to look up $M=O(n^{1/c})$ buckets but with this scheme, one needs to look up $L=O(n^{2/c})$ offsets, it makes the network inefficiency issue even more severe.

\section{Distributed LSH}
\label{sec:distlsh}

In this section, we will present the Layered LSH scheme and theoretically analyze it. We will focus on the $d$-dimensional Euclidian space under $l_2$ norm. As notation, we will let $S$ to be a set of $n$ data points available a-priori, and $Q$ to be the set of query points, either given as a batch (in case of MapReduce) or arriving in real-time (in case of Active DHT). Parameters $k, L, W$ and LSH families $\mathcal{H}=\mathcal{H}_W$ and $\mathcal{H'}=\mathcal{H'}_W$ will be as defined in section \ref{sec:prelim}. Since multiple hash tables can be obviously implemented in parallel, for the sake of clarity we will focus on a single hash table and use a randomly chosen hash function $H\in \mathcal{H'}$ as our LSH function throughout the section. 

%As mentioned in section \ref{sec:bckgrnd}, we will consider the distributed (Key, Value) based systems, and as concrete instantiations focus on the batched processing model MapReduce and the real-time processing model Active DHT. 

In (Key, Value) based distributed systems, a hash function from the domain of all Keys to the domain of available machines is implicitly used to determine the machine responsible for each (Key, Value) pair. In this section, for the sake of clarity, we will assume this mapping to be simply identity. That is, the machine responsible for a (Key, Value) data element is simply the machine with id equal to Key.

At the core, Layered LSH is a carefully distributed implementation of Entropy LSH. Hence before presenting it, first we further detail the simple distributed implementation of Entropy LSH, described in section \ref{sec:bckgrnd}, and explain its major drawback. For any data point $p\in S$ a (Key, Value) pair $(H(p),p)$ is generated and sent to machine $H(p)$. For each query point $q$, after generating the offsets $q+\delta_i$ ($1\leq i\leq L$), for each unique value $x$ in the set \[\{H(q+\delta_i)| 1\leq i \leq L)\},\] a (Key, Value) pair $(x,q)$ is generated and sent to machine $x$. Hence, machine $x$ will have all the data points $p\in S$ with $H(p)=x$ as well as all query points $q\in Q$ such that $H(q+\delta_i)=x$ for some $1\leq i\leq L$. Then, for any received query point $q$, this machine retrieves all data points $p$ with $H(p)=x$ which are within distance $cr$ of q, if any such data points exist. This is done via a UDF in Active DHT or the Reducer in MapReduce, as presented in Figure \ref{distribution:H} for the sake of concreteness of exposition.  

In this implementation, the network load due to data points is not very significant. Not only just one (Key, Value) pair per data point is transmitted over the network, but also in many real-time applications, data indexing is done offline when efficiency and speed are not as critical. However, the amount of data transmitted per query in this implementation is $O(Ld)$: $L$ (Key, Value) pairs, one per offset, each with the $d$-dimensional point $q$ as Value. Both $L$ and $d$ are large in many practical applications with high-dimensional data (e.g., $L$ can be in the hundreds, and $d$ in the tens or hundreds). Hence, this implementation needs a lot of network communication per query, and with a large batch of queries or with queries arriving in real-time at very high rates, this will not only put a lot of strain on the valuable and usually shared network resources but also significantly slow down the search process.

Therefore, a distributed LSH scheme with significantly better query network efficiency is needed. This is where Layered LSH comes into the picture.

%In this paper we focus our attention of distributed frameworks based on (Key, Value) abstraction. Two important instantiations of such frameworks are MapReduce (with its open source implementation Hadoop \cite{Hadoop}), which is widely used in batch processing in high throughput high latency applications, and Active Distributed Hash Tables like Yahoo!'s S4 \cite{S4} and Twitter's Storm \cite{Storm} which are used for real-time processing in low latency applications.  Automatic handling of low level issues like fault tolerance and ease of programming has led to widespread use of distributed frameworks based on the (Key, Value) abstraction relative to traditional message passing based systems. 

%\subsection{Distributed Entropy LSH}
%In this section we explore distributed implementations of {\bf Entropy LSH} both MapReduce and Active DHT frameworks. For the sake of clarity, in this section we consider a single hash table implementation of {\bf Entropy LSH} [and mention how our methods can be extended to indexing schemes with multiple hash tables]. Consider a data set $S$ and a query set $Q$ consisting of points in $\mathbb{R}^d$. Let parameters $k,L,W$ be chosen as defined above in Section 2. Let $H$ denote random choice of a hash function from the LSH family $\mathcal{H}_W$. Also, let $q_1 \ldots q_L$ denote randomly chosen offsets in surface of $B(q,R)$ for each $q \in Q$. [{\bf This subsection needs to be synchronized with notation in Section 2}].

\begin{figure*}[ttt]

 \begin{minipage}[t]{0.45\textwidth}
 \begin{algorithm}[H]
\caption{MapReduce Implementation} 

\begin{algorithmic}
\STATE {\bf Map:}
\STATE{{\bf Input:} Data set $S$, query set $Q$}
\STATE{Choose $H$ from $\mathcal{H'}_W$ uniformly at random, but consistently across Mappers}
\FOR{each data point $p\in S$} 
\STATE Emit $(H(p), p)$
\ENDFOR
\FOR{each query point $q\in Q$} 
\FOR{$1\leq i\leq L$}
\STATE Choose the offset $q+\delta_i$ from the surface of $B(q,r)$
\STATE Emit $(H(q+\delta_i), q)$
\ENDFOR
\ENDFOR
\STATE{}
\STATE {\bf Reduce:}
\STATE{{\bf Input}: For a hash bucket $x$, all data points $p\in S$ with $H(p)=x$, and all query points $q\in Q$ one of whose offsets hashes to $x$. }
\FOR{each query point $q$ among the input points} 
\FOR{each data point $p$ among the input points}
\IF {$p$ is within distance $cr$ of $q$}
\STATE Emit $(q,p)$
\ENDIF
\ENDFOR
\ENDFOR
\end{algorithmic}

\end{algorithm}
\end{minipage}
\hspace{1cm}
\begin{minipage}[t]{0.45\textwidth}
\begin{algorithm}[H]
\caption{Active DHT Implementation} 
\label{H:ADHT} 
\begin{algorithmic}
\STATE {\bf Preprocessing:}
\STATE{{\bf Input:} Data set $S$}
\FOR{each data point $p\in S$}
\STATE{Compute the hash bucket $y=H(p)$}
\STATE{Send the pair $(y, p)$ to machine with id $y$}
\STATE{At machine $y$ add $p$ to the in-memory bucket $y$} 
\ENDFOR
\STATE{}
\STATE {\bf Query Time:}
\STATE{{\bf Input:} Query point $q\in Q$ arriving in real-time} 
\FOR{$1\leq i \leq L$}
\STATE{Generate the offset $q+\delta_i$}
\STATE{Compute the hash bucket $x=H(q+\delta_i)$}
\STATE{Send the pair $(x,q)$ to machine with id $x$}
\STATE{At machine $x$, run $\text{SearchUDF}(x,q)$}
\ENDFOR
\STATE{}
\STATE{\bf $\text{SearchUDF}(x,q)$:}
\FOR{each data point $p$ with $H(p)=x$}
\IF{$p$ is within distance $cr$ of $q$}
\STATE{Emit $(q,p)$}
\ENDIF
\ENDFOR
\end{algorithmic}
\end{algorithm}
\end{minipage}

\caption{Simple Distributed LSH}
\label{distribution:H}
\end{figure*}

%\subsection*{H-scheme}
%We first describe a simple and natural distributed implementation of LSH which we briefly outlined in Section 1. This approach involves associating the pair $(H(s),s)$ with each $s \in S$ and the pairs $(H(q_i),q)$, $i = 1 \ldots L$ for each $q \in Q$. These (Key, Value) pairs are randomly distributed over the network in a manner which guarantees that all pairs with same Key are mapped to the same node (or Reduce task in MapReduce).  Each node (Reduce Task in MapReduce) then performs a local search of data points mapped to it in order to identify $cR$-neighbors to all the query points it received. This search can be implemented via a UDF in Active DHT, the Reducer in MapReduce. We describe this mechanism  compactly in figure(\ref{distribution:H}) and refer to it as the $H$-scheme. Correctness of implementation in both the frameworks is guaranteed by the property that (Key, Value) pairs with the same Key are sent to the same node (or Reduce task). Also, it can be easily seen that per query, $O(L)$ amount of data needs to be shuffled over the network. Appropriately setting the seeds used for random number generation can ensure that all the Map tasks generate the correct hash function $H$ and query offsets $q_i,1\leq i \leq L$ for each query $q \in Q$. 

\subsection{Layered LSH}
%\subsection*{G-scheme} 

In this subsection, we present the Layered LSH scheme. The main idea is to use another layer of locality sensitive hashing to distribute the data and query points over the machines. More specifically, given a parameter value $D>0$, we sample an LSH function $G:\mathbb{R}^k \to \mathbb{Z}$ such that:

\begin{equation}\label{eq:G}
G(v) = \lfloor \frac{\alpha \cdot v + \beta }{D} \rfloor
\end{equation}

where ${\bf \alpha} \in \mathbb{R}^k$ is a $k$-dimensional vector whose individual entries are chosen from the standard Gaussian $\mathcal{N}(0,1)$ distribution, and $\beta\in \mathbb{R}$ is chosen uniformly from $[0,D]$.

Then, denoting $G(H(\cdot))$ by $GH(\cdot)$, for each data point $p\in S$, we generate a (Key, Value) pair $(GH(p), <H(p),p>)$, which gets sent to machine $GH(p)$. By breaking down the Value part to its two pieces, $H(p)$ and $p$, this machine will then add $p$ to the bucket $H(p)$. This can be done by the Reducer in MapReduce, and by a UDF in Active DHT. Similarly, for each query point $q\in Q$, after generating the offsets $q+\delta_i$ ($1\leq i\leq L$), for each unique value $x$ in the set
\begin{equation}\label{eq:Gset}
 \{GH(q+\delta_i)| \, 1\leq i \leq L\}
 \end{equation}
we generate a (Key, Value) pair $(x, q)$ which gets sent to machine $x$. Then, machine $x$ will have all the data points $p$ such that $GH(p)=x$ as well as the queries $q\in Q$ one of whose offsets gets mapped to $x$ by $GH(\cdot)$. Specifically, if for the offset $q+\delta_i$, we have $GH(q+\delta_i)=x$, all the data points $p$ that $H(p)=H(q+\delta_i)$ are also located on machine $x$. Then, this machine regenerates the offsets $q+\delta_i$ ($1\leq i\leq L$), finds their hash buckets $H(q+\delta_i)$, and for any of these buckets such that $GH(q+\delta_i)=x$, it performs a similarity search among the data points in that bucket. Note that since $q$ is sent to this machine, there exists at least one such bucket. Also note that, the offset regeneration, hash, and bucket search can all be done by either a UDF in Active DHT or the Reducer in MapReduce. To make the exposition more concrete, we have presented the pseudo code for both the MapReduce and Active DHT implementations of this scheme in Figure \ref{distribution:G}. 

At an intuitive level, the main idea in Layered LSH is that since $G$ is an LSH, and also for any query point $q$, we have $H(q+\delta_i)\simeq H(q)$ for all offsets $q+\delta_i$ ($1\leq i\leq L$), the set in equation \ref{eq:Gset} has a very small cardinality, which in turn implies a small amount of network communication per query. On the other hand, since $G$ and $H$ are both LSH functions, if two data points $p,p'$ are far apart, $GH(p)$ and $GH(p')$ are highly likely to be different. This means that, while locating the nearby points on the same machines, Layered LSH partitions the faraway data points on different machines, which in turn ensures a good load balance across the machines. Note that this is critical, as without a good load balance, the point in distributing the implementation would be lost.

In the next section, we present the formal analysis of this scheme, and prove that compared to the simple implementation, it provides an exponential improvement in the network traffic, while maintaining a good load balance across the machines.

\begin{figure*}[ht]
\begin{minipage}[t]{0.45\linewidth}
\begin{algorithm}[H]
\caption{MapReduce Implementation} 

\begin{algorithmic}
\STATE {\bf Map:}
\STATE{{\bf Input:} Data set $S$, query set $Q$}
\STATE{Choose hash functions $H,G$ randomly but consistently across mappers}
	\FOR{each data point $p\in S$} 
		\STATE{ Emit $(GH(p), <H(p), p>)$}
	\ENDFOR
	\FOR{each query point $q\in Q$}
		\FOR{$1\leq i\leq L$}
			\STATE{Generate the offset $q+\delta_i$}
			\STATE{Emit $(GH(q+\delta_i),q)$} 
		\ENDFOR		 
	\ENDFOR
\STATE{}  
\STATE{{\bf Reduce:}}
	\STATE{{\bf Input}:  For a hash bucket $x$, all pairs $<H(p),p>$ for data points $p\in S$ with $GH(p)=x$, and all query points $q\in Q$ one of whose offsets is mapped to $x$ by $GH$.}
	\FOR{each data point $p$ among the input points} 
		\STATE{Add $p$ to bucket $H(p)$} 
	\ENDFOR  
	\FOR{each query point $q$ among the input points}
		\FOR{$1\leq i\leq L$}
			\STATE{Generate the offset $q+\delta_i$, and find $H(q+\delta_i)$}
			\IF{$GH(q+\delta_i)=x, H(q+\delta_i)\neq H(q+\delta_j)$ ($\forall j<i$)}
				\FOR{each data point $p$ in bucket $H(q+\delta_i)$}
					\IF{($p$ is within distance $cr$ of $q$)}
						\STATE{Emit $(q,p)$}
					\ENDIF
				\ENDFOR	
			\ENDIF 
		\ENDFOR
	\ENDFOR 
\end{algorithmic}
\end{algorithm}
\end{minipage}
\hspace{1cm}
\begin{minipage}[t]{0.45\linewidth}
\begin{algorithm}[H]
\caption{Active DHT Implementation} 
\label{H:ADHT} 
\begin{algorithmic}
\STATE {\bf Preprocessing:}
\STATE{{\bf Input:} Data set $S$}
\FOR{each data point $p\in S$}
\STATE {Compute the hash bucket $H(p)$ and machine id $y=GH(p)$}
\STATE{Send the pair $(y, <H(p),p>)$ to machine with id $y$}
\STATE{At machine $y$, add $p$ to the in-memory bucket $H(p)$} 
\ENDFOR
\STATE{}
\STATE {\bf Query Time:}
\STATE{{\bf Input:} Query point $q\in Q$ arriving in real-time}
\FOR{$1\leq i\leq L$}
\STATE{Generate the offset $q+\delta_i$, compute $x=GH(q+\delta_i)$}
\IF{$GH(q+\delta_j)\neq x \, (\forall j<i)$}
\STATE{Send the pair $(x,q)$ to machine with id $x$}
\STATE{At machine $x$, run $\text{SearchUDF}(x,q)$}
\ENDIF
%\STATE{Return $NN(GoH(q_i),(H(q_i),q))$, where $NN$: UDF returns $cR$-neighbors of $q$ in bucket $H(q_i)$ at location $GoH(q_i)$}
\ENDFOR
\STATE{}
\STATE{{\bf $\text{SearchUDF}(x,q)$:}}
\FOR{$1\leq i\leq L$}
\STATE{Generate offset $q+\delta_i$, compute $H(q+\delta_i), GH(q+\delta_i)$}
\IF{$GH(q+\delta_i)=x, H(q+\delta_i)\neq H(q+\delta_j) \, (\forall j<i)$}
\FOR {each data point $p$ in bucket $H(q+\delta_i)$}
\IF{$p$ is within distance $cr$ from $q$}
\STATE {Emit $(q,p)$}
\ENDIF
\ENDFOR
\ENDIF
\ENDFOR
\end{algorithmic}
\end{algorithm}
\end{minipage}
\caption{Layered LSH}
\label{distribution:G}
\end{figure*}

\subsection{Analysis}
\label{sec:analysis}

In this section, we analyze the Layered LSH scheme presented in the previous section. We first fix some notation. As mentioned earlier in the paper, we are interested in the $(c,r)$-NN problem. Without loss of generality and to simplify the notation, in this section we assume $r=1/c$. This can be achieved by a simple scaling.
The LSH function $H\in \mathcal{H'}_W$ that we use is $H=(H_1, \ldots, H_k)$, where $k$ is chosen as in Theorem \ref{thm:ELSH} and for each $1\leq i \leq k$:
\[
H_i(v) = \lfloor \frac{a_i\cdot v + b_i}{W} \rfloor
\] 
where $a_i$ is a $d$-dimensional vector each of whose entries is chosen from the standard Gaussian $\mathcal{N}(0,1)$ distribution, and $b_i \in \mathbb{R}$ is chosen uniformly from $[0,W]$. We will also let $\Gamma:\mathbb{R}^d\rightarrow \mathbb{R}^k$ be $\Gamma = (\Gamma_1, \ldots, \Gamma_k)$, where for $1\leq i \leq k$:
\[
\Gamma_i(v) = \frac{a_i\cdot v + b_i}{W} 
\]
hence, $H_i(\cdot) = \lfloor \Gamma_i(\cdot) \rfloor$. We will use the following small lemma in our analysis:
\begin{lemma} \label{lem:trieq} For any two vectors $u,v\in \mathbb{R}^d$, we have:
\[
||\Gamma(u)-\Gamma(v)|| - \sqrt{k} \leq ||H(u)-H(v)|| \leq ||\Gamma(u)-\Gamma(v)|| + \sqrt{k} 
\]
\end{lemma}
\begin{proof}
Denoting $R_i=\Gamma_i - H_i$ ($1\leq i\leq k$) and $R=(R_1, \ldots, R_k)$, we have $0\leq R_i(u), R_i(v) \leq 1$ ($1\leq i \leq k$), and hence $||R(u)-R(v)|| \leq \sqrt{k}$. Also, by definition $H=\Gamma - R$, and hance $H(u)-H(v) = (\Gamma(u)-\Gamma(v)) + (R(v)-R(u))$. Then, the result follows from triangle inequality.
\end{proof}
Our analysis also uses two well-known facts. The first is the sharp concentration of $\chi^2$-distributed random variables, which is also used in the proof of the Johnson-Lindenstrauss lemma \cite{im98, Dasgupta_JL}, and the second is the $2$-stability property of Gaussian distribution:

\begin{fact} \label{fact:sharpcon} If $\omega\in \mathbb{R}^m$ is a random $m$-dimensional vector each of whose entries is chosen from the standard Gaussian $\mathcal{N}(0,1)$ distribution, and $m=\Omega(\frac{\log n}{\epsilon^2})$, then with probability at least $1-\frac{1}{n^{\Theta(1)}}$, we have $$(1-\epsilon)\sqrt{m} \leq ||\omega|| \leq (1+\epsilon)\sqrt{m}$$
\end{fact}

\begin{fact} \label{fact:2stable} If $\theta$ is a vector each of whose entries is chosen from the standard Gaussian $\mathcal{N}(0,1)$ distribution, then for any vector $v$ of the same dimension, the random variable $ \theta \cdot v$ has Gaussian $\mathcal{N}(0, ||v||)$ distribution.
\end{fact}

The plan for the analysis is as follows. We will first analyze (in theorem \ref{thm:fq}) the network traffic of Layered LSH and derive a formula for it based on $D$, the bin size of LSH function $G$. We will see that as expected, increasing $D$ reduces the network traffic, and our formula will show the exact relation between the two. We will next analyze (in theorem \ref{thm:rc}) the load balance of Layered LSH and derive a formula for it, again based on $D$. Intuitively speaking, a large value of $D$ tends to put all points on one or few machines, which is undesirable from the load balance perspective. Our analysis will formulate this dependence and show its exact form. These two results together will then show the exact tradeoff governing the choice of $D$, which we will use to prove (in Corollary \ref{cor:Dsqrtk}) that with an appropriate choice of $D$, Layered LSH achieves both network efficiency and load balance. Before proceeding to the analysis, we give a definition:
\begin{defn}
\label{defn:f} Having chosen LSH functions $G,H$, for a query point $q\in Q$, with offsets $q+\delta_i$ ($1\leq i\leq L$), define
$$f_q = |\{GH(q+\delta_i)| 1\leq i\leq L\}|$$
to be the number of (Key, Value) pairs sent over the network for query $q$.
\end{defn}

Since $q$ is $d$-dimensional, the network load due to query $q$ is $O(df_q)$. Hence, to analyze the network efficiency of Layered LSH, it suffices to analyze $f_q$. This is done in the following theorem:
\begin{thm} \label{thm:fq} For any query point $q$, with high probability, that is probability at least $1-\frac{1}{n^{\Theta(1)}}$, we have: $$f_q=O(\frac{k}{D})$$
\end{thm}
\begin{proof}
Since for any offset $q+\delta_i$, the value $GH(q+\delta_i)$ is an integer, we have: 
\begin{equation} \label{eq:fqmax}
f_q \leq \max_{1\leq i,j \leq L} \left\{GH(q+\delta_i)-GH(q+\delta_j)\right\}
\end{equation}
For any vector $v$, we have: $$\frac{\alpha.v+\beta}{D} -1 \leq G(v) \leq \frac{\alpha.v+\beta}{D}$$ hence for any $1\leq i,j\leq L$: $$GH(q+\delta_i)-GH(q+\delta_j) \leq \frac{\alpha\cdot (H(q+\delta_i)-H(q+\delta_j))}{D}+1$$
Thus, from equation \ref{eq:fqmax}, we get:
\[
f_q \leq \frac{1}{D} \max_{1\leq i,j \leq L} \left\{\alpha\cdot (H(q+\delta_i)-H(q+\delta_j)) \right\}+1
\]
From Cauchy-Schwartz inequality for inner products, we have for any $1\leq i,j\leq L$:
\[
\alpha\cdot (H(q+\delta_i)-H(q+\delta_j)) \leq ||\alpha|| \cdot ||H(q+\delta_i)-H(q+\delta_j)||
\]
Hence, we get:
\begin{equation}
\label{eq:fqnorms}
f_q \leq \frac{||\alpha||}{D} \cdot \max_{1\leq i,j \leq L} \left\{||H(q+\delta_i)-H(q+\delta_j)|| \right\}+1
\end{equation}
For any $1\leq i,j\leq L$, we know from lemma \ref{lem:trieq}: 
\begin{equation}
\label{eq:HGamma}
||H(q+\delta_i)-H(q+\delta_j)|| \leq ||\Gamma(q+\delta_i)-\Gamma(q+\delta_j)|| + \sqrt{k}
\end{equation}
Furthermore for any $1\leq t\leq k$, since
$$\Gamma_t(q+\delta_i) - \Gamma_t(q+\delta_j) = \frac{a_t\cdot (\delta_i-\delta_j)}{W}$$
we know, using Fact \ref{fact:2stable}, that $\Gamma_t(q+\delta_i) - \Gamma_t(q+\delta_j)$ is distributed as Gaussian $\mathcal{N}(0,\frac{||\delta_i-\delta_j||}{W})$. Now, recall from theorem \ref{thm:ELSH} that for our LSH function, $k\geq \frac{\log n}{\log (1/p_2)}$. Hence, there is a constant $\epsilon = \epsilon(p_2) <1$ for which we have, using Fact \ref{fact:sharpcon}: 
\[
||\Gamma(q+\delta_i) - \Gamma(q+\delta_j)|| \leq (1+\epsilon)\sqrt{k}\frac{||\delta_i-\delta_j||}{W}
\]
with probability at least $1-1/n^{\Theta(1)}$. Since, as explained in section \ref{sec:prelim}, all offsets are chosen from the surface of the sphere $B(q,1/c)$ of radius $1/c$ centered at $q$, we have: $||\delta_i-\delta_j||\leq 2/c$. Hence overall, for any $1\leq i,j\leq k$:
\[
||\Gamma(q+\delta_i) - \Gamma(q+\delta_j)|| \leq 2(1+\epsilon)\frac{\sqrt{k}}{cW} \leq 4\frac{\sqrt{k}}{cW}
\]
with high probability. Then, since there are only $L^2$ different choices of $i,j$, and $L$ is only polynomially large in $n$, we get:
\[
\max_{1\leq i,j\leq L} \{||\Gamma(q+\delta_i) - \Gamma(q+\delta_j)||\} \leq 4\frac{\sqrt{k}}{cW}
\]
with high probability. Then, using equation \ref{eq:HGamma}, we get with high probability:
\begin{equation}
\label{eq: Hbound}
\max_{1\leq i,j\leq L} \{||H(q+\delta_i) - H(q+\delta_j)||\} \leq (1+\frac{4}{cW})\sqrt{k}
\end{equation}
Furthermore, since each entry of $\alpha\in \mathbb{R}^k$ is distributed as $\mathcal{N}(0,1)$, another application of Fact \ref{fact:sharpcon} gives (again with $\epsilon=\epsilon(p_2)$): 
\begin{equation}
\label{eq:alphabound}
||\alpha||\leq (1+\epsilon)\sqrt{k} \leq 2\sqrt{k}
\end{equation}
with high probability. Then, equations \ref{eq:fqnorms}, \ref{eq: Hbound}, and \ref{eq:alphabound} together give:
\[
f_q \leq 2(1+\frac{4}{cW})\frac{k}{D}+1
\]
which finishes the proof. 
\end{proof}
\begin{rem} A surprising property of Layered LSH demonstrated by theorem \ref{thm:fq} is that the network load is independent of the number of query offsets, $L$. Note that with Entropy LSH, to increase the search quality one needs to increase the number of offsets, which will then directly increase the network load. Similarly, with basic LSH, to increase the search quality one needs to increase the number of hash tables, which again directly increases the network load. However, with Layered LSH the network efficiency is achieved independently of the level of search quality. Hence, search quality can be increased without any effect on the network load! 
\end{rem}
Next, we proceed to analyzing the load balance of \algo. First, recalling the classic definition of error function:
\[
\mbox{erf}(z) = \frac{2}{\sqrt{\pi}}\int_{0}^{z} e^{-\tau^2} d\tau
\]
we define the function $P(\cdot)$:
\begin{equation}\label{eq:P}
P(z) = \mbox{erf(z)} - \frac{1}{\sqrt{\pi}z}(1-e^{-z^2})
\end{equation}
and prove the following lemma:
\begin{lemma} For any two points $u,v\in \mathbb{R}^k$ with $||u-v||=\lambda$, we have:
\label{lem:GueqGv}
\[
Pr[G(u)=G(v)] = P(\frac{D}{\sqrt{2}\lambda})
\]
\end{lemma}
\begin{proof} Since $\beta$ is uniformly distributed over $[0,D]$, we have:
\[
Pr[G(u)=G(v) |\, \alpha \cdot (u-v) = l] = \max\left\{0, 1-\frac{|l|}{D}\right\} 
\]
Then, since by Fact \ref{fact:2stable}, $\alpha \cdot (u-v)$ is distributed as Gaussian $\mathcal{N}(0,\lambda)$, we have:
\begin{align*}
Pr[G(u)=G(v)] &= \int_{-D}^{D} (1-\frac{|l|}{D})\frac{1}{\sqrt{2\pi}\lambda}e^{-\frac{l^2}{2\lambda^2}}dl\\
&= 2\int_{0}^{D} (1-\frac{l}{D})\frac{1}{\sqrt{2\pi}\lambda}e^{-\frac{l^2}{2\lambda^2}}dl\\
&= \int_{0}^{D} \frac{\sqrt{2}}{\sqrt{\pi}\lambda}e^{-\frac{l^2}{2\lambda^2}}dl - \int_{0}^{D} \frac{l\sqrt{2}}{D\lambda\sqrt{\pi}}e^{-\frac{l^2}{2\lambda^2}}dl\\
&= \mbox{erf}(\frac{D}{\sqrt{2}\lambda}) - \sqrt{\frac{2}{\pi}}\frac{\lambda}{D}(1-e^{-\frac{D^2}{2\lambda^2}}) \\
&= P(\frac{D}{\sqrt{2}\lambda})
\end{align*}
\end{proof}

One can easily see that $P(\cdot)$ is a monotonically increasing function, and for any $0<\xi<1$ there exists a number $z=z_\xi$ such that $P(z_\xi)=\xi$. Using this notation and the previous lemma, we prove the following theorem:

\begin{thm}
\label{thm:rc} For any constant $0<\xi<1$, there is a $\lambda_\xi$ such that
\[
\lambda_\xi/W =  O(1+\frac{D}{\sqrt{k}})
\]
and for any two points $u,v$ with $||u-v||\geq \lambda_\xi$, we have:
\[
Pr[GH(u)=GH(v)] \leq \xi + o(1)
\]
where $o(1)$ is polynomially small in $n$.
\end{thm}
\begin{proof}
Let $u,v\in\mathbb{R}^d$ be two points and denote $||u-v||=\lambda$. Then by lemma \ref{lem:trieq}, we have:
\[
||H(u)-H(v)|| \geq ||\Gamma(u)-\Gamma(v)|| - \sqrt{k} 
\]
As in the proof of theorem \ref{thm:fq}, one can see, using $k\geq \frac{\log n}{\log {(1/p_2)}}$ (from theorem \ref{thm:ELSH}) and Fact \ref{fact:sharpcon}, that there exists an $\epsilon = \epsilon(p_2)=\Theta(1)$ such that with probability at least $1-\frac{1}{n^{\Theta(1)}}$ we have:
\[
||\Gamma(u)-\Gamma(v)|| \geq (1-\epsilon) \frac{\lambda\sqrt{k}}{W}
\]
Hence, with probability at least $1-\frac{1}{n^{\Theta(1)}}$, we have:
\begin{equation*}
%\label{eq:HLowerBound}
||H(u)-H(v)|| \geq \lambda' = ((1-\epsilon) \frac{\lambda}{W} - 1)\sqrt{k} 
\end{equation*}
Now, letting:
\[
\lambda_\xi = \frac{1}{1-\epsilon}(1+\frac{D}{z_\xi \sqrt{2k}})W
\]
we have if $\lambda \geq \lambda_\xi$ then $\lambda' \geq \frac{D}{\sqrt{2}{z_\xi}}$, and hence by lemma \ref{lem:GueqGv}:
\[
Pr[GH(u)=GH(v)|\, ||H(u)-H(v)|| \geq \lambda'] \leq P(\frac{D}{\sqrt{2}\lambda'}) \leq \xi
\]
which finishes the proof by recalling $Pr[||H(u)-H(v)|| < \lambda'] = o(1)$.
\end{proof}

Theorems \ref{thm:fq}, \ref{thm:rc} show the tradeoff governing the choice of parameter $D$. Increasing $D$ reduces network traffic at the cost of more skewed load distribution. We need to choose $D$ such that the load is balanced yet the network traffic is low. Theorem \ref{thm:rc} shows that choosing $D=o(\sqrt{k})$ does not asymptotically help with the distance threshold at which points become likely to be sent to different machines. On the other hand, theorem \ref{thm:rc} also shows that choosing $D=\omega(\sqrt{k})$ is undesirable, as it unnecessarily skews the load distribution. To observe this more clearly, recall that intuitively speaking, the goal in Layered LSH is that if two data point $p_1, p_2$ hash to the same values as two of the offsets $q+\delta_i, q+\delta_j$ (for some $1\leq i,j\leq L$) of a query point $q$ (i.e., $H(p_1)=H(q+\delta_i)$ and $H(p_2)=H(q+\delta_j)$), then $p_1,p_2$ are likely to be sent to the same machine. Since $H$ has a bin size of $W$, such pair of points $p_1,p_2$ most likely have distance $O(W)$. Hence, $D$ should be only large enough to make points which are $O(W)$ away likely to be sent to the same machine. Theorem \ref{thm:rc} shows that to do so, we need to choose $D$ such that:
\[
O(1+\frac{D}{\sqrt{k}}) = O(1)
\]
that is $D=O(\sqrt{k})$. Then, by theorem \ref{thm:fq}, to minimize the network traffic, we choose $D=\Theta(\sqrt{k})$, and get $f_q  = O(\sqrt{k}) = O(\sqrt{\log n})$. This is summarized in the following corollary:

\begin{corol}
\label{cor:Dsqrtk}
Choosing $D=\Theta(\sqrt{k})$, Layered LSH guarantees that the number of (Key,Value) pairs sent over the network per query is $O(\sqrt{\log n})$ with high probability, and yet points which are $\Omega(W)$ away get sent to different machines with constant probability.
\end{corol}

\begin{rem} Corollary \ref{cor:Dsqrtk} shows that, compared to the simple distributed implementation of Entropy LSH and basic LSH, Layered LSH exponentially improves the network load, from $O(n^{\Theta(1)})$ to $O(\sqrt{\log n})$, while maintaining the load balance across the different machines. 
\end{rem}

\section{Experiments}
\label{sec:exp}

\begin{figure*}[ht]
\begin{centering}
\subfigure[Random: recall] { \epsfig{file = 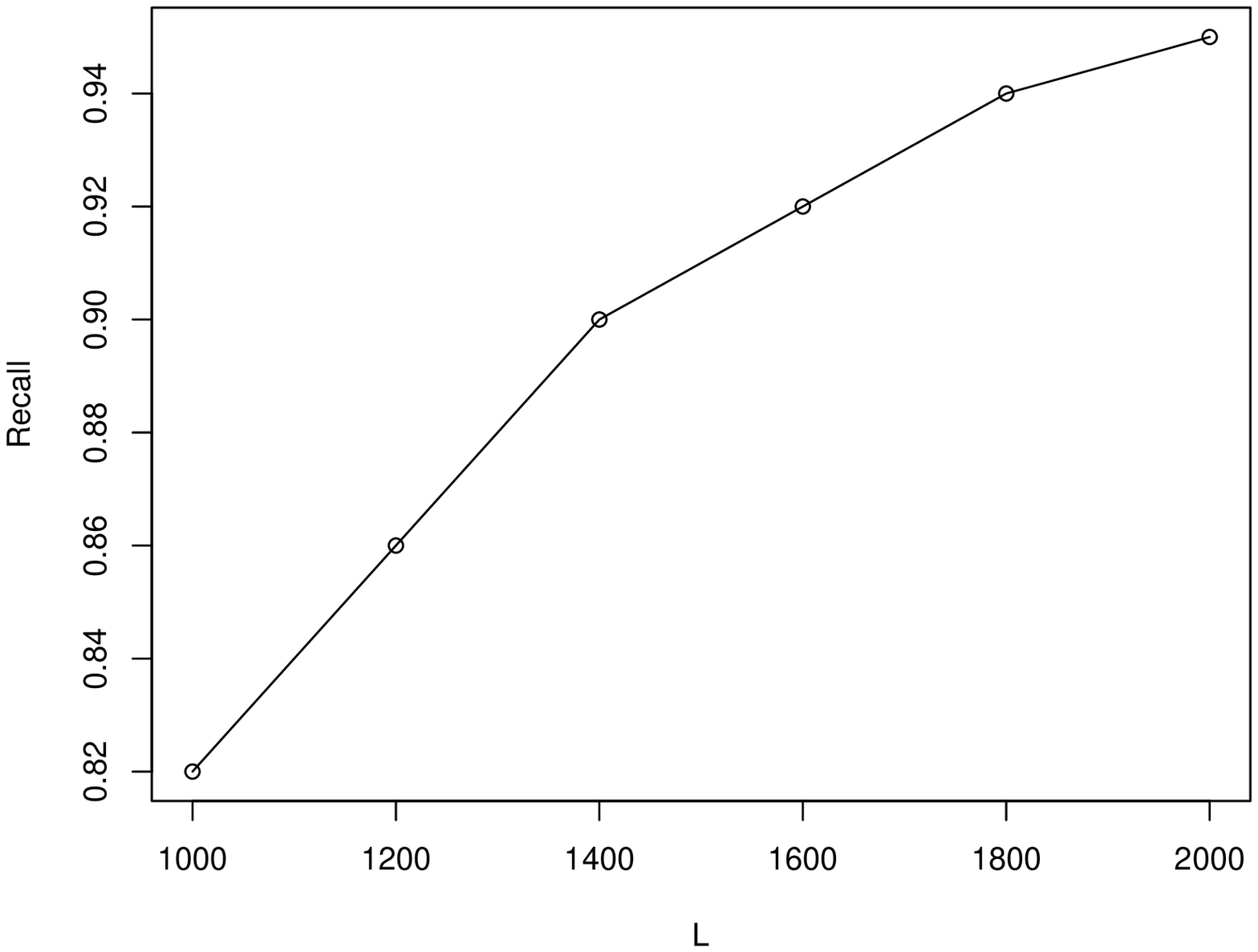, width = 0.3\linewidth, height = 1.4in} \label{Random:recall}} 
\subfigure[Random: shuffle size]{ \epsfig{file = 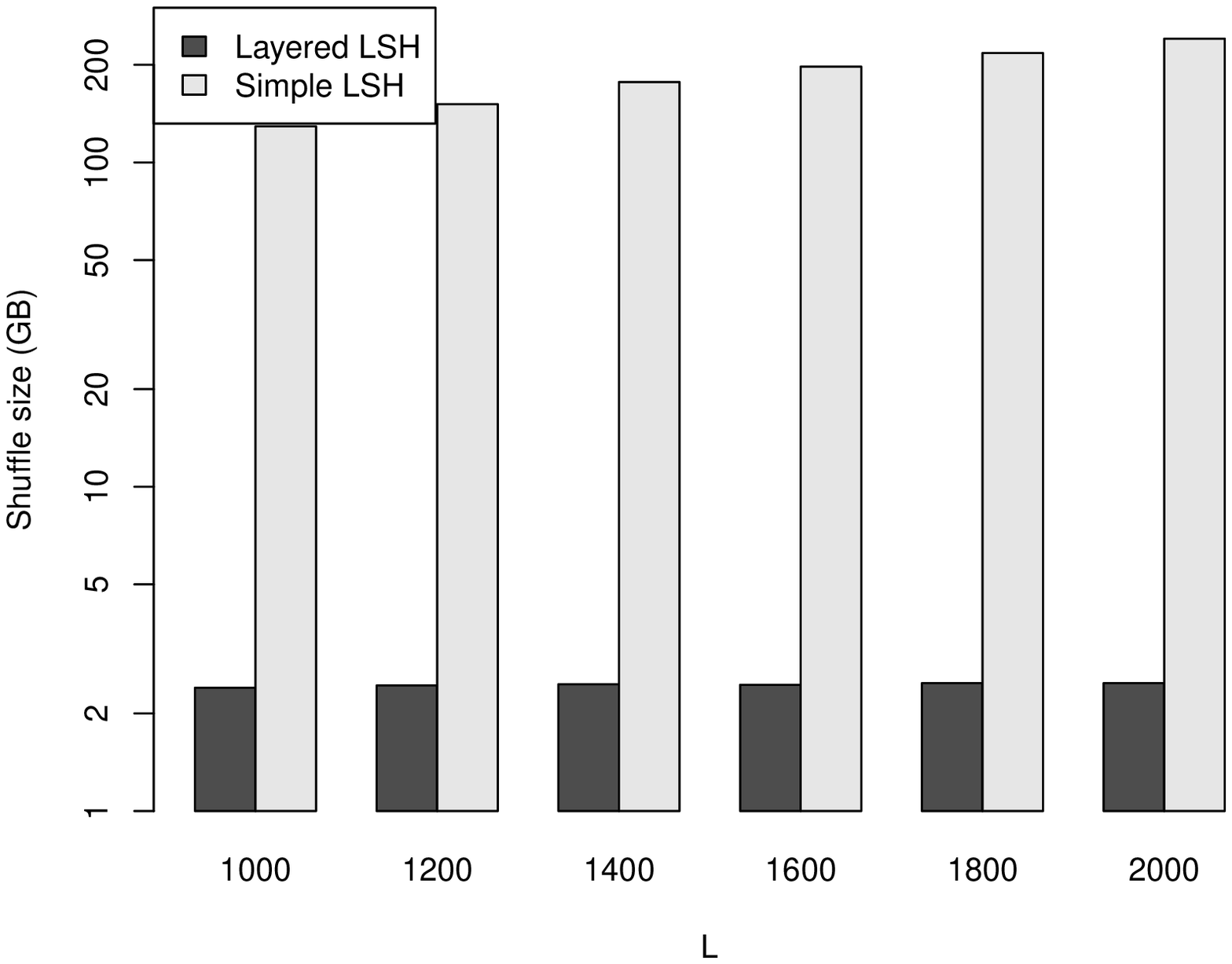, width = 0.3\linewidth, height = 1.4in} \label{Random:shuffle}} 
\subfigure[Random: runtime]{ \epsfig{file = 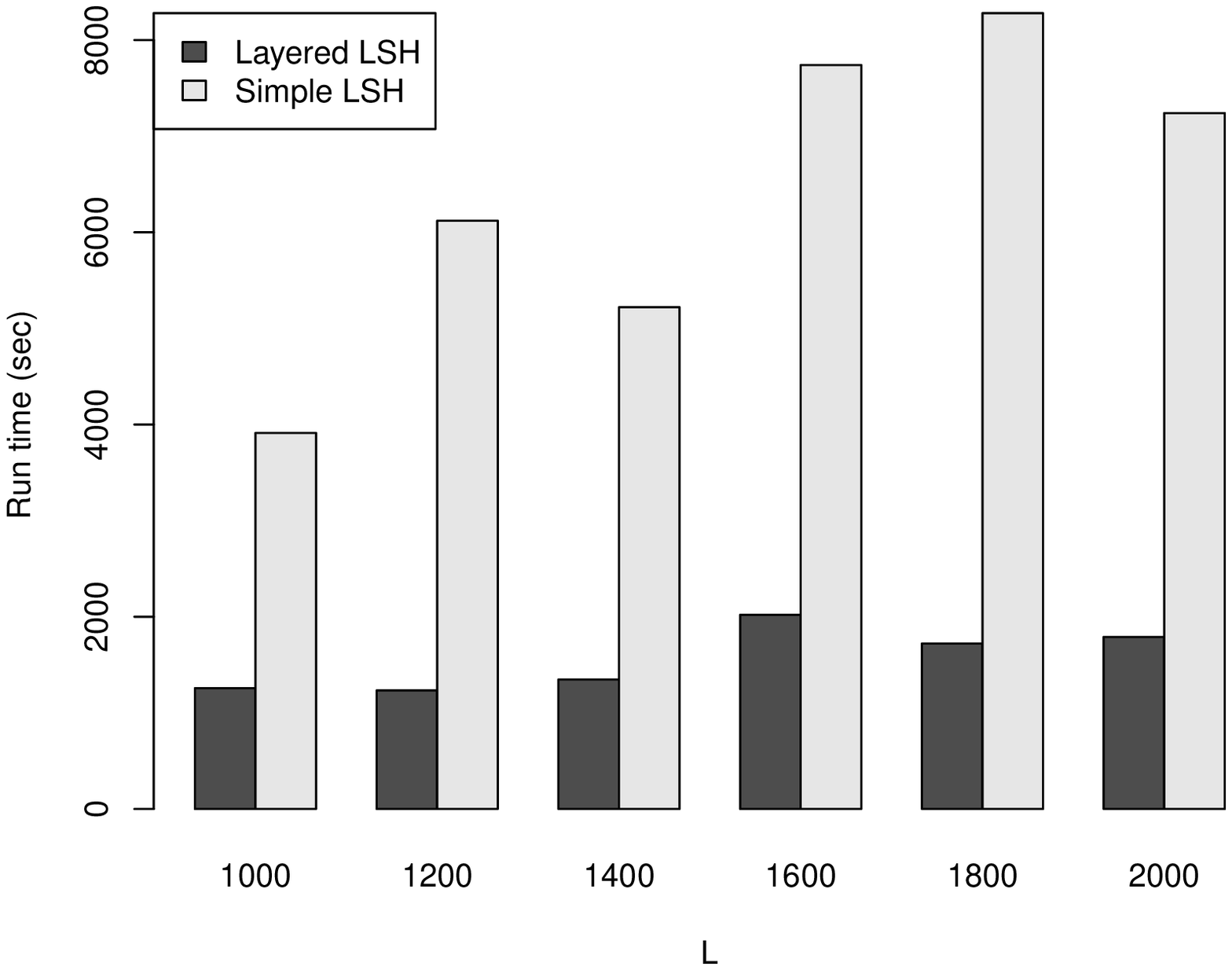, width = 0.3\linewidth, height = 1.4in} \label{Random:runtime}} 
\\
\subfigure[Wiki: recall] { \epsfig{file = 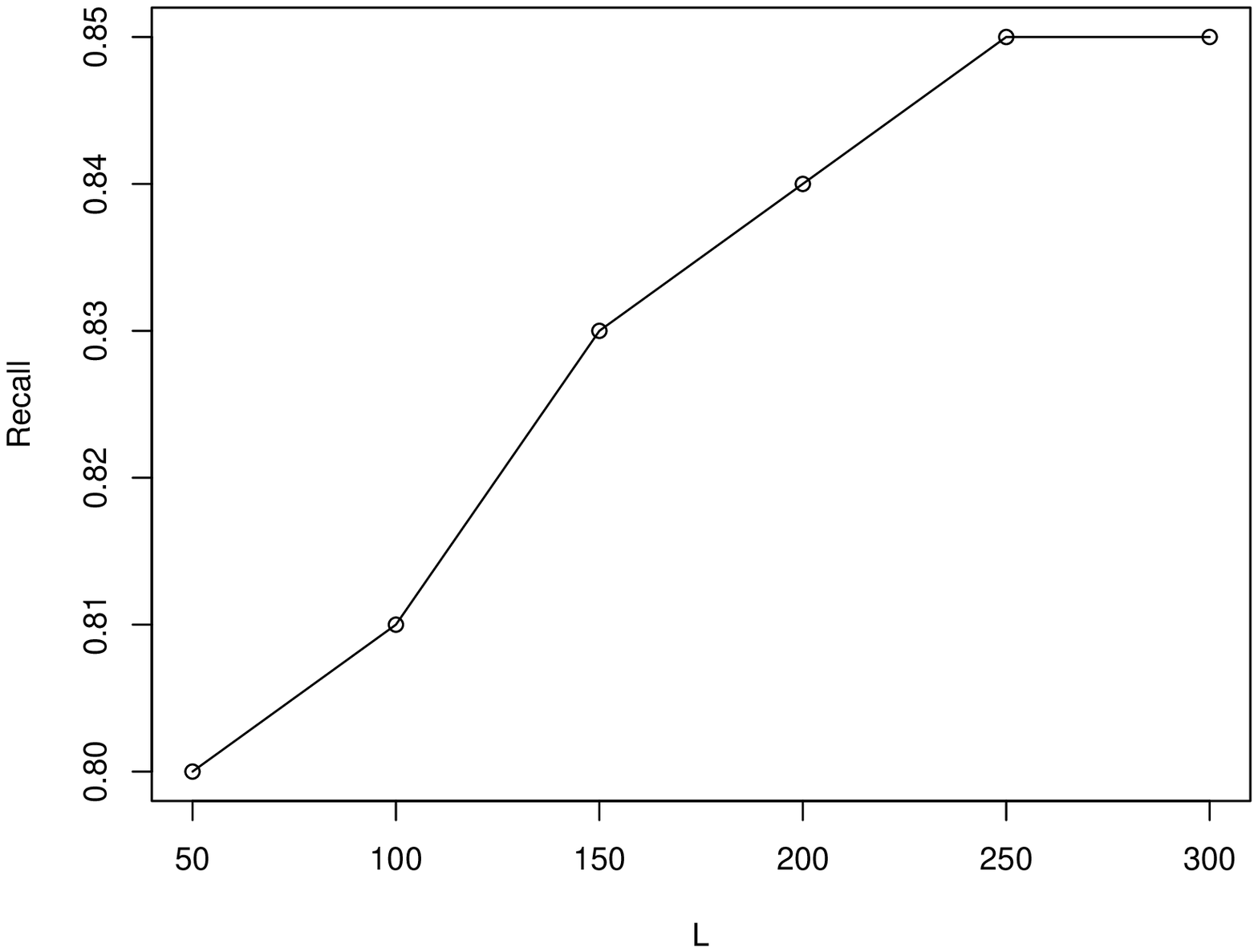, width = 0.3\linewidth, height = 1.4in} \label{Wiki:recall}}
\subfigure[Wiki: shuffle size]{ \epsfig{file = 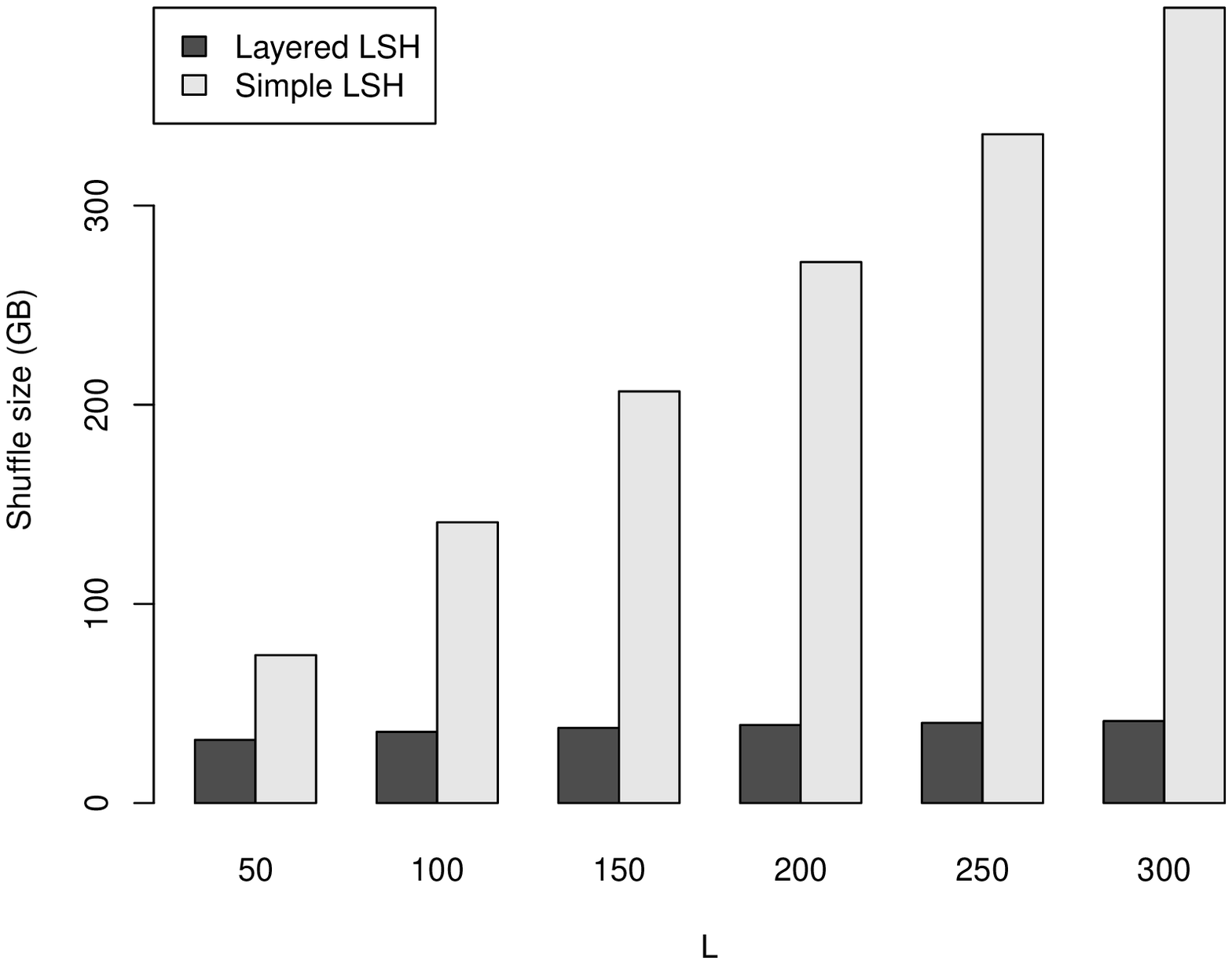, width = 0.3\linewidth, height = 1.4in} \label{Wiki:shuffle}} 
\subfigure[Wiki: runtime]{ \epsfig{file = 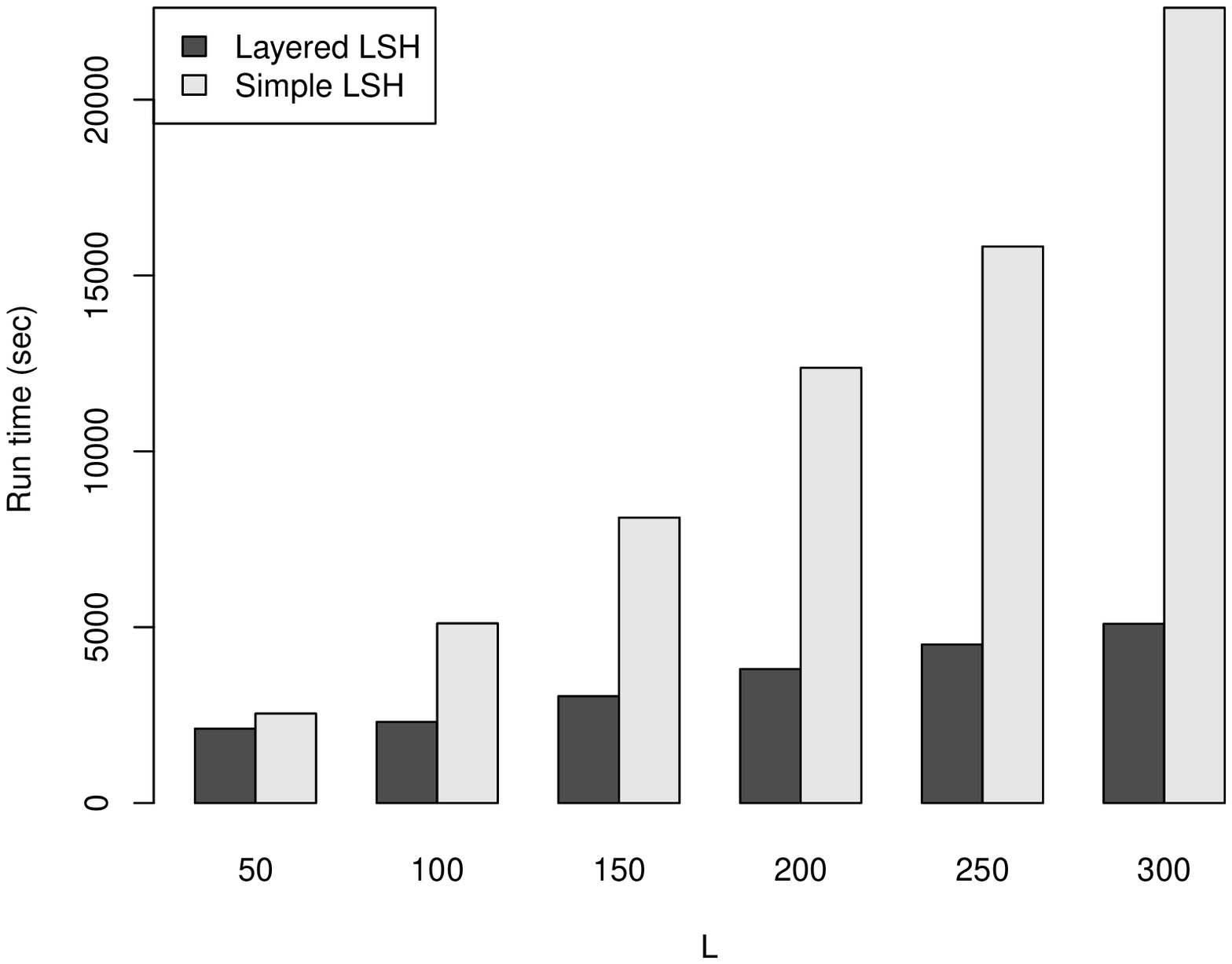, width = 0.3\linewidth, height = 1.4in} \label{Random:runtime}} 
\\
\subfigure[Image: recall] { \epsfig{file = 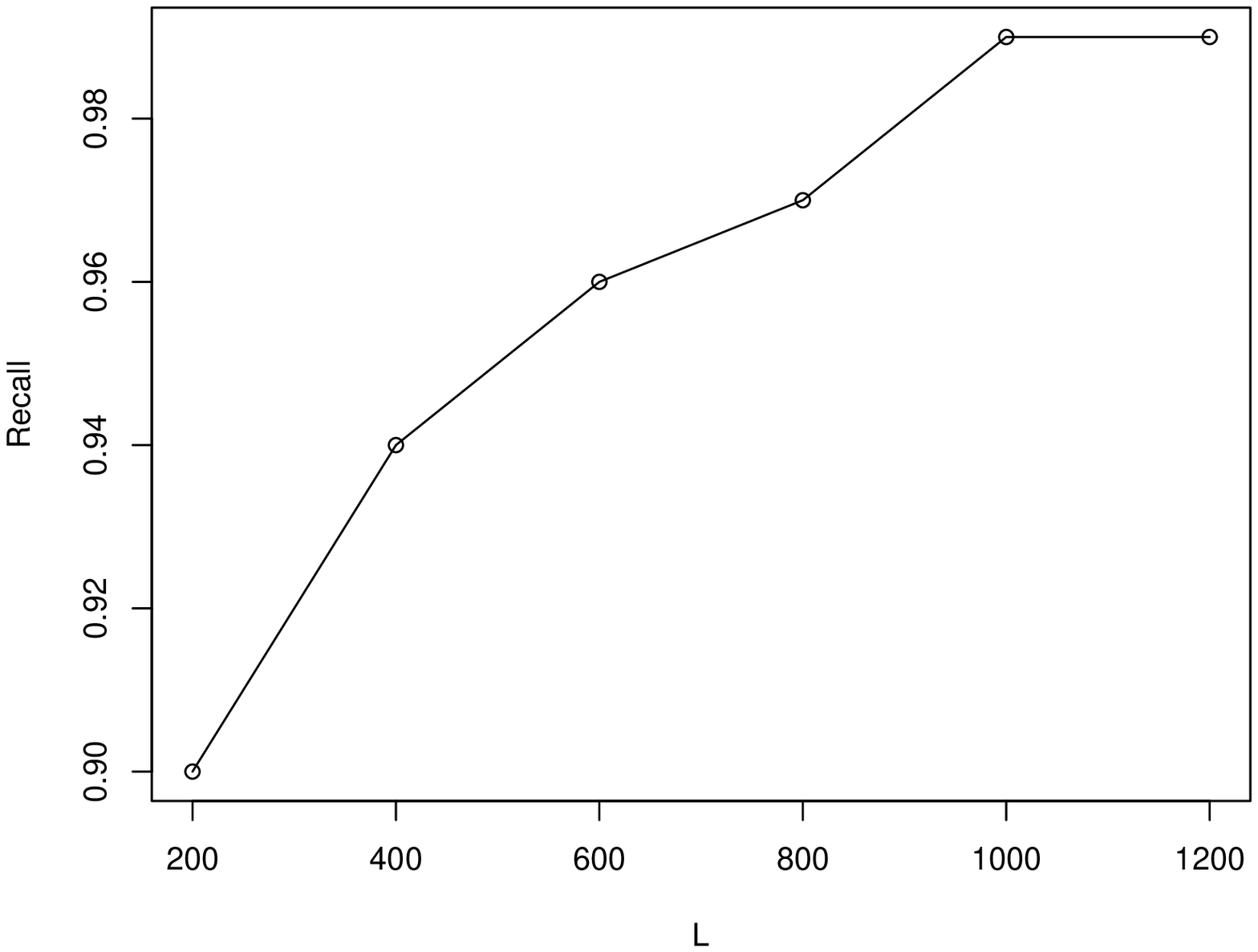, width = 0.3\linewidth, height = 1.4in} \label{Image:recall}} 
\subfigure[Image: shuffle size]{ \epsfig{file = 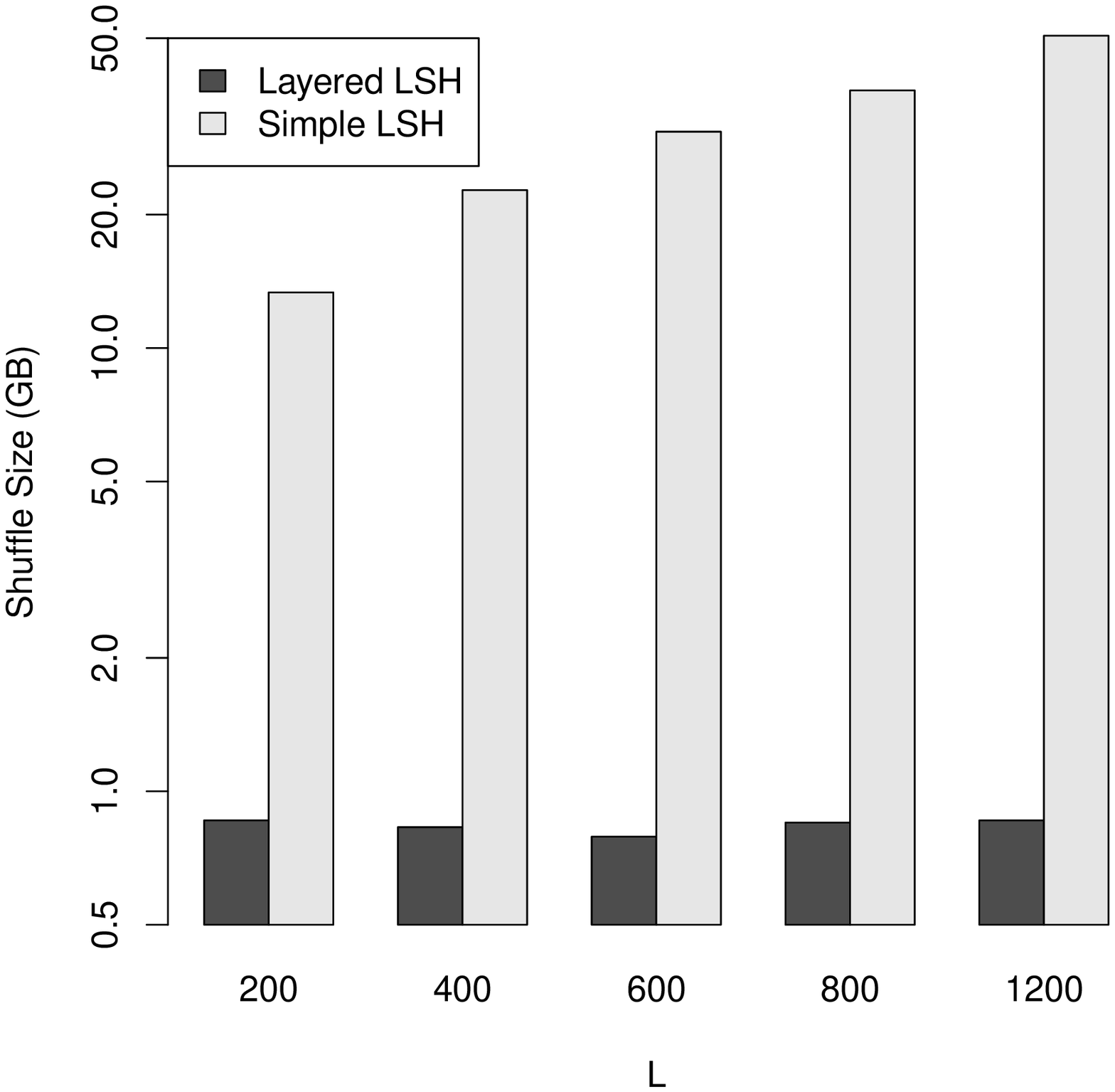, width = 0.3\linewidth, height = 1.4in} \label{Image:shuffle}} 
\subfigure[Image: runtime]{ \epsfig{file = 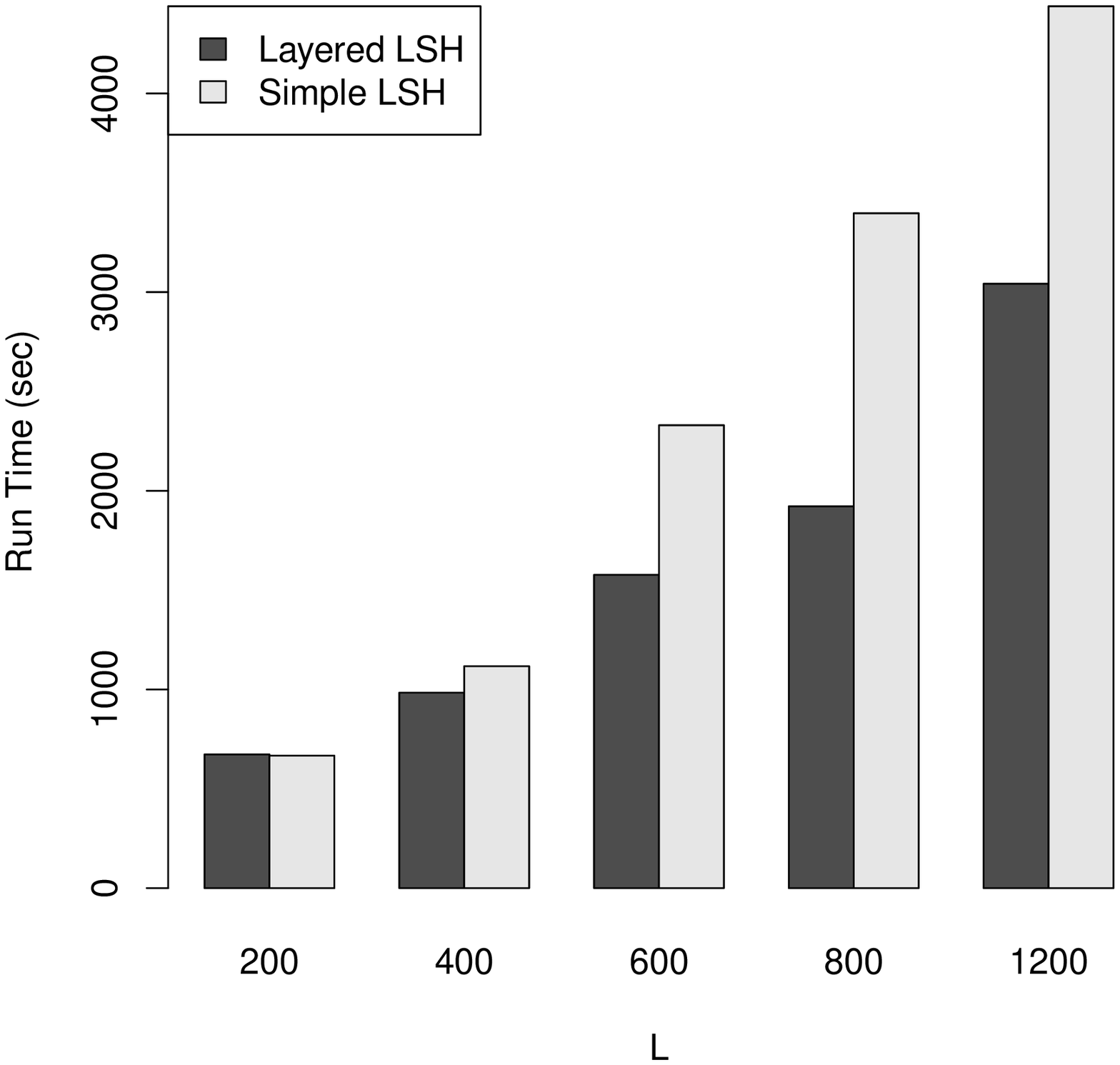, width = 0.3\linewidth, height = 1.4in} \label{Image:runtime}} 
\caption{Variation in recall, shuffle size and wall-clock run time with $L$ for Random, Wiki and Image data sets.}
\label{L}
\end{centering}
\end{figure*}

%\begin{figure*}[ht]
%\begin{center}
%\subfigure[Random: runtime]{ \epsfig{file = pods/random_runtime.eps, width = 0.3\linewidth} \label{Random:runtime}} 
%\subfigure[Image: runtime, $r = 0.08$]{ \epsfig{file = pods/image_runtime_8.eps, width = 0.3\linewidth} \label{Image:runtime:8}} 
%\subfigure[Image: runtime, $r = 0.15$]{ \epsfig{file = pods/image_runtime_15.eps, width = 0.3\linewidth} \label{Image:runtime:16}} 
%\caption{Variation in ``wall-clock" runtime with increasing $L$. Layered LSH yields a on average, a factor $4$ improvement for the Random data set and a factor $2$ improvement for the Image data set. }
%\label{L:runtime}
%\end{center}
%\end{figure*}

In this section, we present an experimental comparison of Simple and Layered LSH via the MapReduce framework with respect to the network cost (shuffle size) and ``wall-clock" run time for a number of data sets. Secondly, we compare Layered LSH in Section~\ref{sec:distlsh} with {\bf Sum} and {\bf Cauchy} distributed LSH schemes described in Haghani et al. \cite{Haghani}. Finally, we also analyze the results by considering the load balance properties.  

\subsection{Datasets} 
First, we describe the data sets we used.   
\begin{itemize}
\item{ \bf Random:}  
This data set is constructed by sampling points from $N^d({\bf 0},1)$\footnote{$N^d({\bf 0},r)$ denotes the normal distribution around the origin, ${\bf 0} \in \mathbb{R}^d$, where the $i$-th coordinate of a randomly chosen point has the distribution $\mathcal{N}(0,r/\sqrt{d})$, $\forall i \in 1 \ldots d$} with $d = 100$ and the queries are generated by adding a small perturbation drawn from $N^d({\bf 0},r)$ to a randomly chosen data point, where $r = 0.3$. We use 1M data points and 100K queries. This ``planted'' data set has been used for LSH experiments in \cite{DIIM04} and we solve the $(c,r)$-NN problem on it with $c = 2$. The parameter choice is such that for each query point, the expected distance to its closest data point is $r$ and that with high probability only that data point is within distance $cr$ from it.

\item{\bf Wiki\footnote{http://download.wikimedia.org}:}
We use the English Wikipedia corpus from February 2012 to compute TF-IDF vectors for each document in it after removing stop words, stemming, and removing insignificant words (appearing fewer than $20$ times in the corpus). We partition the 3.75M articles in the corpus randomly into a data set of size 3M and a query set of size 750K. We solve the $(c,r)$-NN problem with $r = 0.1$ and $c = 2$. 

\item{\bf Image \cite{TinyImages}:} 
The Tiny Image Data set consists of almost $80$M ``tiny" images of size $32\times 32$ pixels \cite{TinyImages}. We extract a 64-dimensional color histogram from each image in this data set using the {\it extractcolorhistogram} tool in the FIRE image search engine, as described in \cite{Charikar:multiprobe, Fire} and normalize it to unit norm in the preprocessing step. 1M Data points and 200K queries are sampled randomly ensuring no overlap. The avg. distance of a query to its closest data point is estimated, through sampling, to be $0.08$ (with standard deviation $0.07$), and hence, we solve the $(c,r)$-NN problem on this data set with $(r = 0.08, c = 2)$.  
\end{itemize}

\comment{
\begin{center}
\begin{tabular}{ l | c | r }
  Data set & No. of data points & No. of queries \\
  \hline                        
  Random & 1M & 100K \\
  Wiki & 3M & 750K \\
  Image & 1M & 200K \\
\end{tabular}
\end{center}
}

\subsection{Implementation Details}

We perform experiments on a small cluster of $13$ compute nodes using Hadoop \cite{Hadoop} with $800$MB JVMs to implement Map and Reduce tasks. Consistency in the choice of hash functions $H,G$ (Section \ref{sec:distlsh}) as well as offsets across mappers and reducers is ensured by setting the seed of the random number generator appropriately. 

We choose the LSH parameters $(W = 0.5, k = 10)$ for the Random data set, $(W = 0.3, k = 16)$ for the Image data set, and $(W = 0.5, k = 12)$ for the Wiki data set according to the calculations in \cite{P06}, and experiments in \cite{Charikar:multiprobe, bayes}. We optimized $D$, the parameter of Layered LSH, using a simple binary search to minimize the wall-clock run time.%while noting that a systematic understanding of the choice of parameters for Layered LSH (as well as Basic LSH) is an important research direction.

Since the underlying dimensionality (vocabulary size $549532$) for the Wiki data set is large, we use the Multi-Probe LSH (MPLSH) \cite{Charikar:multiprobe} as our first layer of hashing for that data set. We discuss MPLSH in detail in Section \ref{sec:rel}. 
We measure the accuracy of search results by computing the recall, i.e., the fraction of query points with at least one data point within a distance $r$, returned in the output. 

\subsection{Results} 
\begin{figure}
\begin{center}
\epsfig{file = 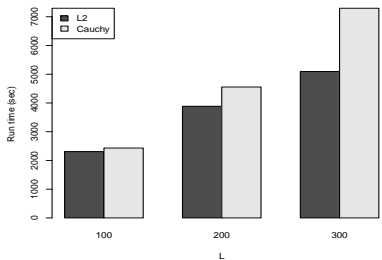, width = 0.65\linewidth, height = 1.6in} 
\caption{Variation in wall-clock run time for the ``Sum'', ``Cauchy'' and Layered LSH schemes with increasing $L$ }
\label{Layer:runtime}
\end{center}
\end{figure}

\begin{enumerate}
\item{\bf Comparison with Simple LSH:} 
Figure \ref{L} describes the results of scaling $L$, the number of offsets per query, on the recall, shuffle size and wall-clock run time using a single hash table. Note that the recall of Entropy LSH can be improved by using $O(1)$ hash tables. Since improving recall is not the main aspect of this paper, we use just a single hash table for all our experiments. We observe that even with a crude binary search for $D$, on average Layered LSH provides a factor 3 improvement over simple LSH in the wall-clock run time on account of a factor 10 or more decrease in the shuffle size. Further, improving recall by increasing $L$ results in a linear increase in the shuffle size for Simple LSH, while, the shuffle size for Layered LSH remains almost constant. This observation verifies Theorem \ref{thm:fq} and Remark $9$. Note that since Hadoop offers checkpointing guarantees where (Key, Value) pairs output by mappers and reducers may be written to the disk, Layered LSH also decreases the amount of data written to the distributed file system. 

\item{\bf Comparison with the ``Sum'' and ``Cauchy'' distributed LSH schemes described in \cite{Haghani}:}
We compare Layered LSH with {\bf Sum} and {\bf Cauchy} distributed LSH schemes described in Haghani et al. \cite{Haghani}. Figure \ref{Layer:runtime} shows that Layered LSH compares favorably with {\bf Cauchy} scheme \footnote{associated parameter chosen via a crude binary search to minimize runtime} for the Wiki data set. The MapReduce job for {\bf Sum} failed due to reduce task running out of memory, indicating load imbalance\footnote{Recall that reduce tasks store data points in memory}. This can be also seen as a manifestation of ``the curse of the last reducer'' \cite{curselast}.

\end{enumerate} 

{\bf Load Balance:}
Next, we discuss the distribution (average and max) of data points in the Wiki data set to $1024$ reduce tasks for the different distribution schemes. 

\begin{table}
	\begin{center}
\begin{tabular}{ l | c | r }
   & Average & Max \\
  \hline                        
  Simple LSH & 3K & 9K \\
  Sum & 3K & 214K \\
  Cauchy & 3K & 45K \\
  Layered LSH & 3K & 103K \\
\end{tabular}
\caption{Wiki data set: distribution of data points across $1024$ reduce tasks}
\label{table:lb}
\end{center}
\end{table}

First, Table \ref{table:lb} above demonstrates {\bf Sum} has the most imbalanced load distribution, explaining its failure on MapReduce. Second, Simple LSH, while having the best load balance for the Wiki data set, incurs a large network cost in order to achieve this load balance. In contrast, Layered LSH offers a tunable way of trading off load balance to decrease network cost and minimize the wall-clock run time. Although {\bf Cauchy} compares favorably to Layered LSH in load balance, it is worse of in running time. In addition, it is not clear if it is possible to provide any theoretical guarantees for the {\bf Cauchy} scheme.

%Network calls (I/O) are often the bottleneck for large scale applications which require to processing huge batches queries, or real time systems where queries arrive at a high rate and experiments presented in this section show that Layered LSH may lead to large improvement in such settings. 

\section{Related Work}
\label{sec:rel}  

\comment{
The methods proposed in the literature for similarity search via nearest neighbors (NN) can be, broadly speaking, classified to data or space partitioning methods and Locality Sensitive Hashing (LSH) techniques. Space partitioning methods, such as $K$-D trees \cite{Bentley:kdtree}, and data partitioning methods, such as R-trees \cite{Rtrees} and SR-trees \cite{SRtrees}, solve the NN problem by pruning the set of candidates for each query using branch and bound techniques. However, they do not scale well with the data dimension, and can in fact be even slower than a brute force sequential scan of the data set when the dimension exceeds $10$ \cite{expt:space:partitioning}. For further details on these methods the reader is referred to the survey by Fodor \cite{Fodor}. 
}

Locality Sensitive Hashing (LSH) was introduced by Indyk and Motwani in order to solve high dimensional similarity search problems \cite{im98}. 
LSH indexing methods are based on LSH families of hash functions for which near points have a higher likelihood of hashing to the same value. Then, $(c,r)$-NN problem can be solved by using multiple hash tables. Gionis et al. \cite{gim99} showed that in the Euclidian space $O(n^{1/c})$ hash tables suffice, which was later improved, by Datar et al. \cite{DIIM04}, to $O(n^{\beta/c})$ (for some $\beta<1$), and further, by Andoni and Indyk \cite{AI06}, to $O(n^{1/{c^2}})$ which almost matches the lower bound proved by Motwani et al. \cite{MNP06}. LSH families are also known for several non-Euclidian metrics, such as Jaccard distance \cite{Broder} and cosine similarity \cite{C02}. 

The main problem with LSH indexing is that to guarantee a good search quality, it requires a large number of hash tables. This entails a large index space requirement, and in the distributed setting, also a large amount of network communication per query. To mitigate the space inefficiency, Panigrahy \cite{P06} proposed Entropy LSH which, by also looking up the hash buckets of $O(n^{2/c})$ random query ``offsets'', requires just $\tilde{O}(1)$ hash tables, and hence provides a large space improvement. But, Entropy LSH does not help with and in fact worsens the network inefficiency of conventional LSH: each query, instead of $O(n^{1/c})$ network calls, one per hash table, requires $O(n^{2/c})$ calls, one per offset. Our Layered LSH scheme exponentially improves this and, while guaranteeing a good load balance, requires only $O(\sqrt{\log n})$ network calls per query.

To reduce the number of offsets required by Entropy LSH, Lv et al. \cite{Charikar:multiprobe} proposed the Multi-Probe LSH (MPLSH) heuristic, in which a query-directed probing sequence is used instead of random offsets. They experimentally show this heuristic improves the number of required offset lookups. In a distributed setting, this translates to a smaller number of network calls per query and Layered LSH can be implemented by using MPLSH instead of Entropy LSH as the first ``layer'' of hashing, as demonstrated by experiments on the Wiki data set in section \ref{sec:exp}. Hence, the benefits of the two methods can be combined in practice.

\comment{
To reduce the number of offsets required by Entropy LSH, Lv et al. \cite{Charikar:multiprobe} proposed the Multi-Probe LSH (MPLSH) heuristic, in which a query-directed probing sequence is used instead of random offsets. They experimentally show this heuristic improves the number of required offset lookups. In a distributed setting, this translates to a smaller number of network calls per query, which is the main goal in Layered LSH as well. Clearly, Layered LSH can be implemented by using MPLSH instead of Entropy LSH as the first ``layer'' of hashing. Hence, the benefits of the two methods can be combined in practice, as demonstrated by experiments on the Wiki data set in section \ref{sec:exp}. However, the experiments by Lv et al. \cite{Charikar:multiprobe} show a modest constant factor reduction in the number of offsets compared to Entropy LSH. Hence, since Layered LSH, besides a theoretical exponential improvement, experimentally shows a factor $100$ reduction in network load compared to Entropy LSH, it is not clear if the marginal benefit from switching the first layer to MPLSH is significant enough to justify its much more complicated offset generation. Furthermore, MPLSH has no theoretical guarantees, while using Entropy LSH as the first layer of hashing allows for the strong theoretical guarantees on both network cost and load balance of Layered LSH proved in this paper. Hence overall, in this paper, we focus on Entropy LSH as the first layer of hashing in Layered LSH.}

Haghani et al. \cite{Haghani} describe the {\bf Sum} and {\bf Cauchy} schemes which map LSH buckets to peers in p2p networks in order to minimize network costs. However, in contrast to Layered LSH, no guarantees on network cost and load balance are provided. In this paper, we show via MapReduce experiments on the Wiki data set that {\bf Sum} distributes data unevenly and thus may load some of the reduce tasks. In addition we also describe experiments which demonstrate that Layered LSH compares favorably with {\bf Cauchy} on this data set.

\section{Conclusions}
\label{sec:conc}

We presented and analyzed Layered LSH, an efficient distributed implementation of LSH similarity search indexing. We proved that, compared to the straightforward distributed implementation of LSH, Layered LSH exponentially improves the network load, while maintaining a good load balance. Our analysis also showed that, surprisingly, the network load of Layered LSH is independent of the search quality. Our experiments confirmed that Layered LSH results in significant network load reductions as well as runtime speedups.  

\comment{
In this paper, we describe a network efficient distributed implementation, Layered LSH, which exponentially reduces network costs relative to the naive implementation while maintaining load balance. Surprisingly, we show that the network cost associated with Layered LSH are independent of the desired accuracy of LSH. 
}

\bibliographystyle{abbrv}
\bibliography{lsh}

\end{document}